\documentclass[%
 aps, physrev,
]{revtex4-2}

\usepackage{graphicx, amsmath, amssymb, amsfonts, amsthm}
\usepackage{dcolumn}
\usepackage{bm}

\newtheorem{theorem}{Theorem}
\newtheorem{lemma}{Lemma}

\theoremstyle{definition}

\newcommand{\outpro}[2]{\vert #1\rangle\langle #2\vert}

\newcommand{\ket}[1]{\vert #1\rangle}

\newcommand{\tr}[1]{\mathsf{Tr}(#1)}

\begin{document}


\title{\textbf{Local Operations in Multiparty Quantum Systems} 
}%

\author{Mithilesh Kumar}
 \email{Contact author: admin@drmithileshkumar.com}
 \homepage{https://drmithileshkumar.com}



\begin{abstract}
In a multipartite systems, local operations are conducted by one party and the results are communicated to the other parties. Such models have been studied under the framework of LOCC and SLOCC. In this paper, we study when can an action of one party be simulated by another. We obtain necessary and sufficient conditions for when can a unitary action be simulated in a bipartite system. We also show that arbitrary operations can be simulated by any party as long as the given multipartite state is Schmidt decomposable. Moreover, we obtain condition for simulation of local measurements in arbitrary tripartite systems.
\end{abstract}

\maketitle

Distributed quantum systems can be seen as multipartite states. Since the early days of quantum information, communication between parties has been studied. Each party does local quantum operation on its part of the system and communicates the result with other parties. Lo and Popescu \cite{Lo} showed that two-way communication between two parties can be reduced to one-way communication. This observation was later used by Nielsen \cite{NielsenLOCC} to obtain the necessary and sufficient condition for conversion of one bipartite state to another via local operations and classical communication. 

The focus of this work is understanding how a local operation done by one party can be simulated by another party. Such a question becomes relevant even in circuit design where operations can be localized to a given set of qubits.

Let us begin with a bipartite system where Alice has one part and Bob has another. Suppose Bob performs local unitary operation \(U_B\). Does there exist local unitary operation \(U_A\) that Alice can perform such that the resulting state is the same? As the following theorem shows, it is not always possible to achieve it. \(U_B\) depends on the input state \(\ket{\psi}\).
\begin{theorem}
    Given a bipartite state \(\ket{\psi}\), the local unitary operation \(U_B\) of Bob can be simulated by a local unitary operation \(U_A\) by Alice if and only if \(U_A = U_B^T\) and \(U_B = U_1\oplus U_2\oplus\cdots \oplus U_d\) where unitary operations \(U_i\) have dimension equal to the multiplicity of the \(i\)th eigenvalue \(\mu_i\) of \(\rho_B\). 
\end{theorem}
\begin{proof}
    Consider the state \(\ket{\psi}\) in the Schmidt basis
    \begin{align*}
        \ket{\psi} &= \sum_\ell\lambda_\ell \ket{\ell_A}\ket{\ell_B}
    \end{align*}
    The state after Bob applies his local unitary operation \(U_B\), written in the Schmidt basis as \begin{align*}
        U_B = \sum_{n,\ell} u^B_{n\ell}\outpro{n_B}{\ell_B}
    \end{align*}
    is given by
    \begin{align*}
        \ket{\phi} &= \sum_\ell \lambda_\ell \ket{\ell_A}U_B\ket{\ell_B}\\
        &= \sum_\ell \lambda_\ell \ket{\ell_A}\sum_n u^B_{\ell n}\ket{n_B}\\
        &= \sum_{\ell,n} \lambda_\ell u^B_{\ell n} \ket{\ell_A}\ket{n_B}\\
        &= \sum_{ij} \lambda_i u^B_{ij}\ket{i_A}\ket{j_B}
    \end{align*}
    where the last expression is just relabeling \(\ell\) to \(i\) and \(n\) to \(j\). Suppose there exists a local unitary operation \(U_A\) for Alice that results in the same state \(\ket{\phi}\) when acted up on \(\ket{\psi}\). Writing \(U_A\) is the Schmidt basis we get
    \begin{align*}
        U_A &= \sum_{m\ell}u^A_{m\ell}\outpro{m_A}{\ell_A}
    \end{align*}
    Applying \(U_A\) on \(\ket{\psi}\) gives
    \begin{align*}
        \ket{\phi} &= \sum_\ell \lambda_\ell U_A\ket{\ell_A}\ket{\ell_B}\\
        &= \sum_\ell \lambda_\ell \sum_m u^A_{\ell m}\ket{m_A}\ket{\ell_B}\\
        &= \sum_{m,\ell} \lambda_\ell u^A_{\ell m} \ket{m_A}\ket{\ell_B}\\
        &= \sum_{ij} \lambda_j u^A_{ji}\ket{i_A}\ket{j_B}
    \end{align*}
    These to separate expressions for \(\ket{\phi}\) imply that for each \(i,j\)
    \begin{align*}
        \lambda_i u^B_{ij} &= \lambda_j u^A_{ji}\\
        v^A_{ij}\lambda_j &= \lambda_i u^B_{ij}
    \end{align*}
    where \(v^A_ij = u^A_{ji}\) is unitary as well. We can define a diagonal matrix \(D\) such that the diagonals are the coefficients \(\lambda_i\). This provides us with
    \begin{align*}
        V_AD &= DU_B\\
        V_A &= DU_BD^{-1}
    \end{align*}
    Now we need to apply the condition that \(V_A\) must be unitary.
    \begin{align*}
        V_A^\dagger V_A &= I_A\\
        \implies (D^{-1}U_B^\dagger D)(DU_B D^{-1}) &= I_A\\
        \implies U_B^\dagger D^2 U_B &= D^2\\
        \implies D^2 U_B &= U_B D^2\\
        \implies \lambda_i^2 u^B_{ij} &= \lambda_j u^B_{ij}\\
        \implies (\lambda_i^2 - \lambda_j^2) u^B_{ij} &= 0
    \end{align*}
    Suppose that \(\lambda_i\) are sorted in decreasing order such that there are blocks of equal values \((\mu_1, ...,\mu_1), (\mu_2,...,\mu_2)...,\)\((\mu_d, ..., \mu_d)\).
    The set of indices for which \(\lambda_i\neq \lambda_j\), we must have \(u^B_{ij} = 0\). For set of indices such that \(\lambda_i = \lambda_j\), the square sub-matrix formed by taking the corresponding columns and rows is unitary, that is, \(U_B\) must be a block diagonal where each block is unitary.
    \begin{align*}
        U_B = U_1\oplus U_2\oplus\cdots \oplus U_d
    \end{align*}
    where the dimension of \(U_i\) is equal to multiplicity of \(\mu_i\). It is easy to verify that \(U_A = U_B^T\).
\end{proof}
It is clear from the above proof why similar ideas worked for measurements in bipartite states, but not for unitary operations. Measurements bring a lot more freedom. Next we show that the proof technique of Lo and Popescu \cite{Lo} can be extended to Schmidt decomposable multipartite systems \cite{Schmidt1907, PERES199516, Ac_n_2000, kumarMultipartite}.

\begin{theorem}
    In a Schmidt decomposable multipartite state, any local quantum operation done by one party can be simulated by any party at the cost of some local unitary operations.
\end{theorem}
\begin{proof}
    The proof is in line with one given by Lo and Popescu \cite{Lo}. For simplicity, we consider the case of tripartite state. Start with the Schmidt decomposition of \(\ket{\psi}\).
    \begin{align*}
        \ket{\psi} = \sum_\ell \lambda_\ell \ket{\ell_A}\ket{\ell_B}\ket{\ell_C}
    \end{align*}
    Let the measurement operators of Alice and Bob are represented by \(\{A_j\}\) and \(\{B_j\}\) respectively. Writing these operators in Schmidt bases of their respective spaces gives
    \begin{align*}
        A_j &= \sum_{k\ell}A^j_{k\ell}\outpro{k_A}{\ell_A}\\
        B_j &= \sum_{k\ell}B^j_{k\ell}\outpro{k_B}{\ell_B}
    \end{align*}
    The action of these measurements of \(\ket{\psi}\) will give
    \begin{align*}
        A_j\ket{\psi} &= a \sum_{kl}\lambda_\ell A^j_{k\ell}\ket{k_A}\ket{\ell_B}\ket{\ell_C}\\
        B_j\ket{\psi} &= b \sum_{kl}\lambda_\ell B^j_{k\ell}\ket{\ell_A}\ket{k_B}\ket{\ell_C}
    \end{align*}
    where \(a\) and \(b\) are normalization constants given by 
    \begin{align*}
        a &= \sum_{k\ell}\lambda^2_\ell |A^j_{k\ell}|^2\\
        b &= \sum_{k\ell}\lambda^2_\ell |B^j_{k\ell}|^2
    \end{align*}
    If Bob initiates measurement and Alice tries to simulate, then we can define
    \begin{align*}
        A^j_{k\ell} = B^j_{k\ell}
    \end{align*}
    After the measurement Alice applies the unitary operation \(U_A: \ket{k_A}\to \ket{\ell_A}\) and Bob applies \(U_B: \ket{\ell_B}\to \ket{k_B}\). Other parties do nothing.
\end{proof}
Schimdt decomposition gives a lot of advantage in dealing with local operations. Next, we express any tripartite state in the Schmidt bases obtained via bipartitions of the tripartite state \(\ket{\psi}\).
\begin{theorem}
    Any tripartite state \(\ket{\psi}\) can be decomposed as 
    \begin{align}
        \ket{\psi} = \sum_{\ell,m,n}a_{\ell mn}\ket{\ell_A}\ket{m_B}\ket{n_C}
    \end{align}
    where \(\ket{\ell_A}, \ket{m_B},\ket{n_C}\) are Schmidt vectors in the Schmidt decomposition of bipartitions \(A-BC, B-AC\) and \(AB-C\) respectively.
\end{theorem}
\begin{proof}
    We obtain the Schmidt bases \(\ket{\ell_A}, \ket{m_B}\) and \(\ket{n_C}\) by considering Schmidt decompositions of bipartitions \(A-BC, B-AC\) and \(AB-C\) respectively.
    \begin{align*}
        \ket{\psi} &= \sum_\ell^{r_A} \alpha_\ell \ket{\ell_A}\ket{\ell_{BC}}\\
                   &= \sum_m^{r_B} \beta_m \ket{m_B}\ket{m_{AC}}\\
                   &= \sum_n^{r_C} \gamma_n \ket{n_{AB}}\ket{n_C}
    \end{align*}
    Without loss of generality assume that \(r_A\leq r_Br_C\). Let \(V_{BC} = span(\{\ket{\ell_{BC}}\})\). Then \(dim(V_{BC}) = r_A\) is a subspace of \(\mathcal{H}_B\otimes \mathcal{H}_C\). Similarly \(span(\{\ket{m_Bn_C}\})\) is a subspace of dimension \(r_Br_C\). Since these basis sets are orthonormal, they can be extended to full basis and can be related via a unitary transformation. This implies that the states \(\ket{\ell_{BC}}\) can be written in the linear combination of states \(\ket{m_Bn_C}\). This implies that we can write
    \begin{align*}
        \ket{\psi} = \sum_{\ell, m, n}a_{\ell mn}\ket{\ell_A}\ket{m_B}\ket{n_C}
    \end{align*}
\end{proof}
The following lemma relates the reduced density matrices of individual parties with the Schmidt bases.
\begin{lemma}
    If \(\ket{\psi} = \sum_{\ell mn}a_{\ell mn}\ket{\ell_A}\ket{m_B}\ket{n_C}\) where \(\ket{\ell_A}, \ket{m_B}\) and \(\ket{n_C}\) are Schmidt bases, then the following holds
    \begin{enumerate}
        \item \(\rho_A = \sum_\ell \alpha_\ell^2 \outpro{\ell_A}{\ell_A}\)
        \item \(AA^\dagger = D_A\) where \(A_{\ell, mn} = a_{\ell mn}\) and \(D_A\) is diagonal with coefficients \(\alpha_\ell^2\)
    \end{enumerate}
    Similar relations hold for other bipartitions.
\end{lemma}
\begin{proof}
    Considering the Schmidt decomposition of bipartition \(A-BC\) we get
    \begin{align*}
        \ket{\psi} &= \sum_\ell \alpha_\ell \ket{\ell_A}\ket{\ell_{BC}}
    \end{align*}
    Tracing out \(\ket{\ell_{BC}}\) gives
    \begin{align}\label{eq:rhoa}
        \rho_A &= \sum_{\ell}\alpha_\ell^2\outpro{\ell_A}{\ell_A}
    \end{align}
    Starting out with \(\ket{\psi} = \sum_{\ell mn}a_{\ell mn}\ket{\ell_A}\ket{m_B}\ket{n_C}\), it is easy to verify that the matrix representation is given by 
    \begin{align}\label{eq:rhoaa}
        \rho_A &= AA^\dagger
    \end{align}
    where the matrix \(A\) is obtained such that \([A]_{\ell, mn} = a_{\ell mn}\). Since we have used same basis in Equation \ref{eq:rhoa} and Equation \ref{eq:rhoaa}, we must have that 
    \begin{align*}
        AA^\dagger = D_A
    \end{align*}
    where \(D_A\) is diagonal with diagonal entries given by \(\alpha_\ell^2\).
\end{proof}
Now we consider the following situation for tripartite systems of Alice, Bob and Cat. Bob makes a measurement \(M_j\) on his system and communicates the results to Alice and Cat. Can Alice make a measurement \(L_j\) instead such that up to local unitary operations, we end up with the same state, i.e. Alice simulates operations of Bob?
\begin{theorem}
    For a tripartite system, given a set of complete measurement operators \(\{M_j\}\) of Bob, Alice can simulate \(\{M_j\}\) if there exists complete measurement operators \(\{L_j\}\) such that 
    \begin{align}
        fL_j = h[M_j A^{T_{AB}}]^{T_{AB}}A^\dagger D_A^{-1}
    \end{align}
    where \begin{itemize}
        \item the matrix \(A = [a_{\ell, mn}]\)
        \item the matrix \(D_A\) is diagonal matrix of eigenvalues of \(\rho_A\)
        \item the matrix \(A^{T_{AB}}\) is obtained by taking partial transpose of indices for Alice and Bob
        \item \(h^{-2} = \sum_{km}\beta_m^2|M^j_{km}|^2\) is normalization factor after taking measurement \(M_j\) where \(\beta_m^2\) are eigenvalues of \(\rho_B\)
        \item \(f^{-2} =  \sum_{d\ell}\alpha_\ell^2|L^j_{d\ell}|^2\) is normalization factor after taking measurement \(L_j\) where \(\alpha_\ell^2\) are eigenvalues of \(\rho_A\)
    \end{itemize}
\end{theorem}
\begin{proof}
    We start by expressing \(\ket{\psi}\) in Schmidt bases of each subsystem 
    \begin{align*}
        \ket{\psi} = \sum_{\ell, m,n} a_{\ell mn}\ket{\ell_A}\ket{m_B}\ket{n_C}
    \end{align*}
    The measurement operators are also expressed in these bases as
    \begin{align*}
        L_j &= \sum_{d\ell}L^j_{d\ell}\outpro{d_A}{\ell_A}\\
        M_J &= \sum_{km}M^j_{km}\outpro{k_B}{m_B}
    \end{align*}
    Applying these measurements to \(\ket{\psi}\), we get
    \begin{align*}
        L_j\ket{\psi} &= f \sum_{mnd}(\sum_\ell L^j_{d\ell} a_{\ell mn})\ket{d_Am_Bn_C}\\
        M_j\ket{\psi} &= h \sum_{\ell kn}(\sum_m M^j_{km} a_{\ell mn})\ket{\ell_A k_Bn_C}
    \end{align*}
    where \(f\) and \(h\) are normalization constants. Relabeling \(d \to \ell\) and \(m\to k\) we can rewrite \(L_j\ket{\psi}\) as 
    \begin{align*}
        L_j\ket{\psi} &= f \sum_{\ell k n}(\sum_p L^j_{\ell p} a_{p kn})\ket{\ell_Ak_Bn_C}
    \end{align*}
    Equating coefficients implies
    \begin{align}\label{eq:measure}
        f\sum_{p}L^j_{\ell p} a_{p\ell n} &= h\sum_m M^j_{km}a_{\ell mn}\\
        f[L_j A]_{\ell, kn} &= h[M_j B]_{k,\ell n}
    \end{align}
    where matrix \(A\) is of dimension \(r_A\times r_Br_C\) and is obtained by considering 
    \begin{align*}
        [A]_{\ell, mn} = a_{\ell mn}
    \end{align*}
    The matrix \(B\) is of dimension \(r_B\times r_Ar_C\) and is obtained by considering
    \begin{align*}
        [B]_{m,\ell n} = a_{\ell mn}
    \end{align*}
    That is, \(B\) can be obtained by taking partial transpose of \(A\) by swapping indices \(\ell \leftrightarrow m\)
    \begin{align*}
        B = A^{T_{AB}}
    \end{align*}
    Returning back to Equation \ref{eq:measure}, we can write
    \begin{align*}
        fL_j A &= h (M_j B)^{T_{AB}}\\
        fL_jAA^\dagger &= h (M_j B)^{T_{AB}}A^\dagger\\
        fL_jD_A &= h (M_j B)^{T_{AB}}A^\dagger\\
        fL_j &= h (M_j B)^{T_{AB}}A^\dagger D_A^{-1}\\
             &= h (M_j A^{T_{AB}})^{T_{AB}}A^\dagger D_A^{-1}
    \end{align*}
    where \(D_A\) is the matrix of diagonal of eigenvalues of \(\rho_A\).

    All that remains now is to obtain the expressions for the normalization factors \(f\) and \(h\). As seen above, the coefficients in \(L_j\ket{\psi}\) are given by \(f[L_jA]_{\ell, kn}\). Using the normalization condition, we get
    \begin{align*}
        \sum_{\ell mn}f^2 [L_jA]_{\ell, kn}[L_jA]^*_{\ell, kn} &= 1\\
        f^2 \sum_{\ell mn} [L_jA]_{\ell, kn} [L_jA]^\dagger_{kn,\ell} &= 1\\
        f^2 \tr{L_jAA^\dagger L_j^\dagger} &= 1\\
        f^2 \tr{L_JD_AL_j^\dagger} &= 1\\
        f^2 \sum_{d\ell}\alpha^2_\ell |L_{d\ell}^2|^2 &= 1
    \end{align*}
    Similarly, we can obtain the normalization constant \(h\) as 
    \begin{align*}
        h^2 \sum_{km}\beta^2_m|M^j_{km}|^2 &= 1
    \end{align*}
\end{proof}

We conclude that in general unitary operations can not be simulated even in bipartite systems, even though it is known that measurements can be simulated (\cite{Lo}). Schmidt decomposable multipartite states allow simulation of arbitrary local operations. We obtain condition for when arbitrary tripartite states allow simulation of measurements.

\bibliography{main}

\end{document}